\documentclass[copyright,creativecommons]{eptcs}
\providecommand{\event}{DCM 2011} 
\usepackage{breakurl}
\usepackage{xcolor}
\usepackage{epsfig}
\usepackage{amsfonts,amsthm}
\usepackage{latexsym,amsmath,amssymb}
\usepackage{prooftree}
\usepackage{listings}
\usepackage{enumerate}

\newcommand{\qqop}[1]{\mathrel{\makebox[2em]{$#1$}}}

\newif\ifsubmit
\submitfalse
\submittrue
\ifsubmit
\newcommand{\AT}[1]{#1}
\newcommand{\ATComm}[1]{}
\newcommand{\MD}[1]{#1}
\newcommand{\MDComm}[1]{}
\newcommand{\PG}[1]{#1}
\newcommand{\PGComm}[1]{}
\newcommand{\LB}[1]{#1}
\newcommand{\LBComm}[1]{}
\else
\newcommand{\AT}[1]{\textcolor{blue}{#1}}
\newcommand{\ATComm}[1]{{\scriptsize \textcolor{blue}{[Angelo{:} #1]}}}
\newcommand{\MD}[1]{\textcolor{red}{#1}}
\newcommand{\MDComm}[1]{{\scriptsize \textcolor{red}{[Mariangiola{:} #1]}}}
\newcommand{\PG}[1]{\textcolor{magenta}{#1}}
\newcommand{\PGComm}[1]{{\scriptsize \textcolor{magenta}{[Paola{:} #1]}}}
\newcommand{\LB}[1]{\textcolor{darkgreen}{#1}}
\newcommand{\LBComm}[1]{{\scriptsize \textcolor{darkgreen}{[Livio{:} #1]}}}\fi
\definecolor{darkgreen}{rgb}{0,0.5,0}

\newcommand{\pat}{pattern}

\newcommand{\DNA}{{\sf DNA}}
\newcommand{\GENE}{{\it g}}
\newcommand{\Rg}{{\it R_g}}
\newcommand{\RmR}{{\it R_m}}
\newcommand{\Rmic}{{\it R_a}}
\newcommand{\Rmu}{{\it R_{m\uparrow}}}
\newcommand{\Rpdo}{{\it R_{o\downarrow}}}
\newcommand{\Rpdi}{{\it R_{i\downarrow}}}
\newcommand{\Rau}{{\it R_{a\uparrow}}}

\newcommand{\up}[2]{#1^{\uparrow #2}}
\newcommand{\down}[2]{#1^{\downarrow #2}}

\newcommand{\Lr}{L}
\newcommand{\RR}{\mathcal{R}}
\newcommand{\CC}{E\!\!\!\!\!E}

\newcommand{\ltrans}[1]{\xrightarrow{}}
\newcommand{\ttrans}[1]{\rightarrow_{\intercal}}

\newcommand{\set}[1]{\{#1\}}

\newcommand{\EE}{\mathcal{E}}

\newcommand{\genVar}{\chi}

\newcommand{\SALTA}[1]{}

\newcommand{\PP}{\mathcal{P}}

\newcommand{\XX}{\mathcal{X}}
\newcommand{\TV}{\mathcal{TV}}
\newcommand{\SV}{\mathcal{SV}}
\newcommand{\TT}{\mathcal{T}}

\newcommand{\SSq}{\mathcal{S}}

\newcommand{\VV}{\mathcal{V}}
\newcommand{\xx}{\widetilde{x}}
\newcommand{\yy}{\widetilde{y}}
\newcommand{\zz}{\widetilde{z}}

\newcommand{\Ltrans}[1]{\Longrightarrow}
\newcommand{\Ttrans}[1]{\Longrightarrow_\intercal}
\newcommand{\Loop}[1]{\left(#1\right)^{\LB{\circlearrowleft}}}

\newcommand{\into}{\ensuremath{\,\rfloor}\,}
\newcommand{\pipe}{\ensuremath{\;|\;}}
\newcommand{\phole}{\square}

\newcommand{\df}{{\sf d}}
\newcommand{\rf}{{\sf r}}
\newcommand{\s}{{\sf s}}
\newcommand{\ef}{{\sf e}}
\newcommand{\fif}{{\sf i}}
\newcommand{\of}{{\sf o}}

\newcommand{\derSeq}{\vdash_s}
\newcommand{\derPat}{\vdash_p}
\newcommand{\derGR}{\vdash_{gr}}

\newcommand{\TypeSeq}[3]{#1\derSeq #2:#3}
\newcommand{\TypePat}[3]{#1\derPat #2:#3}
\newcommand{\TypeGR}[3]{#1\derGR #2:{\tt ok}}
\newcommand{\GenericVariable}{\chi}
\newcommand{\T}{\Delta}
\newcommand{\La}{\Lambda}

\newcommand{\tSeq}{\varphi}
\newcommand{\tPat}{\tau}

\newcommand{\Rtunion}{\sqcup}
\newcommand{\Rtsub}{\sqsubseteq}
\newcommand{\feature}{{\sf features}}

\newcommand{\emptyseq}{\emptyset}
\newcommand{\conc}[2]{#1\cdot#2}

\newcommand{\upRule}[4]{\up{#1}{#2}\! \mapsto \! \up{#3}{#4}}

\newcommand{\agr}{\quad\big|\quad}

\newcommand{\mycdot}{\!\cdot\!}

\newcommand{\lab}[1]{\ensuremath{\scriptstyle{\textsc{(#1)}}}}

\newcommand{\frTT}{\underline{\TT}}
\newcommand{\freeze}[1]{\underline{#1}}
\newcommand{\defrost}{\eta}
\newcommand{\labl}{{\it l}}

\newtheorem{theorem}{Theorem}[section]
\newtheorem{definition}[theorem]{Definition}

\newtheorem{lemma}[theorem]{Lemma}
\newtheorem{example}[theorem]{Example}

\title{A Calculus of Looping Sequences with Local Rules\thanks{This research is founded by the BioBITs Project (\emph{Converging Technologies} 2007, area: Biotechnology-ICT), Regione Piemonte.}}

\author{Livio Bioglio \institute{Dip. di Informatica, Univ. di Torino, Italy}
\and Mariangiola Dezani-Ciancaglini
\institute{Dip. di Informatica, Univ. di Torino, Italy} \and Paola
Giannini \institute{Dip. di Informatica, Univ. del Piemonte
Orientale, Italy}
 \and Angelo Troina
\institute{Dip. di Informatica, Univ. di Torino, Italy}
 }

\def\titlerunning{A Calculus of Looping Sequences with Local Rules}
\def\authorrunning{L. Bioglio, M. Dezani-Ciancaglini, P. Giannini and A. Troina}
\begin{document}
\maketitle

\begin{abstract}
In this paper we present a variant of the Calculus of Looping
Sequences (CLS for short) with global and local rewrite rules. While
global rules, as in CLS, are applied anywhere in a given term, local
rules can only be applied in the compartment on which they are
defined. Local rules are dynamic: they can be added, moved and
erased. We enrich the new calculus with a parallel semantics where a
reduction step is lead by any number of global and local rules that
could be performed in parallel. A type system is developed to
enforce the property that a compartment must contain only local
rules with specific features. As a running example we model some
interactions happening in a cell starting from its nucleus and
moving towards its mitochondria.
\end{abstract}

\section{Introduction}

The Calculus of Looping Sequences (CLS for
short)~\cite{BarMagMilTro06a,BarMagMilTro06b,Mil07,BMMT07}, is a formalism for describing
biological systems and their evolution. CLS is based on term
rewriting with a set of predefined rules modelling the activities
one would like to describe. CLS terms are constructed by starting from a set of basic constituent elements which are composed with operators of sequencing, looping, containment and parallel composition. Sequences may represent DNA fragments and proteins, looping sequences may represent membranes, parallel composition may represent juxtaposition of elements and populations of chemical species.

The model has been extended with several
features such as bisimulations \cite{BarMagMilTro06b,BarMagMilTro08},
combining the
simplicity of notation of rewrite systems with the advantage of a
form of compositionality.
In~\cite{ADT08,BDMMT10} a type system was defined to ensure the well-formedness of links between protein sites within the Linked Calculus of Looping Sequences (see~\cite{BarMagMil06}). In~\cite{DGT09a} we defined a type system to guarantee the soundness of reduction rules with respect to the requirement of certain elements, and the repellency of others.

In this paper we present a variant of CLS with global and
local rewrite rules (CLSLR, for short). Global rules are
applied anywhere in a given term wherever their patterns
match the portion of the system under investigation, local
rules can only be applied in the compartment in which they
are defined. Terms written in CLSLR are thus syntactically
extended to contain explicit local rules within the term,
on different compartments. Local rules can be created, moved
between different compartments and deleted. \PG{We feel that
having a calculus in which we can model the dynamic
evolution of the rules describing the system
results in a more natural and direct way to study emerging
properties of complex systems.}
As it happens in nature, where \emph{data} and \emph{programs}
are encoded in the same kind of molecular structures, we insert
rewrite rules within the terms modelling the system under investigation.

In CLSLR we also enrich CLS with a parallel semantics in which we define a reduction step lead by any number of
global and local rules that could be performed in parallel.

Since in this framework the focus is put on local rules, we define a
set of \emph{features} that can be associated to each local rule.
Features may define general properties of rewrite rules or
properties which are strictly related to the model under
investigation. We define a \emph{membrane type} for the compartments
of our model and develop a type systems enforcing the property that
a compartment must contain only local rules with specific features.

Thus, the main features of CLSLR are:
\begin{itemize}
\item different compartments can evolve according to different \emph{local} rules;
\item the set of \emph{global} rules is fixed;
\item local rules are \emph{dynamic}: they can be added, moved and erased according to both global and local rules;
\item a \emph{parallel} reduction step permits the application of several global and local rules;
\item compartments are enforced to contain only rules with specific \emph{features}.
\end{itemize}

As a running case study, emphasising the peculiarities of the
calculus, we consider some mitochondrial activity underlining the
form of symbiosis between a cell and its mitochondria
(see~\cite{DLTL06}). Mitochondria are membrane-enclosed organelle
found in eukaryotic cells that generate most of the cell's energy
supply in the form of adenosine triphosphate (ATP). A mitochondrion
is formed by two membranes, the outer and the inner membrane, having
different properties and proteins on their surfaces. Both membranes
have receptors to mediate the entrance of molecules. In
Figure~\ref{FIG_mit} we show the expression of a gene (encoded in
the DNA within the nucleus of the cell) destined to be translated
into a protein that will be catch by mitochondria and will then
catalyse the production of ATP. In particular, we will model the
following steps: (1) genes within the nucleus' DNA are transcribed
into mRNA, (2) mRNA moves from the nucleus to the cell's cytoplasm,
(3) where it is translated into the protein.

\begin{figure}[t]
\begin{center}
\begin{minipage}{0.98\textwidth}
\begin{center}
\includegraphics[width=90mm]{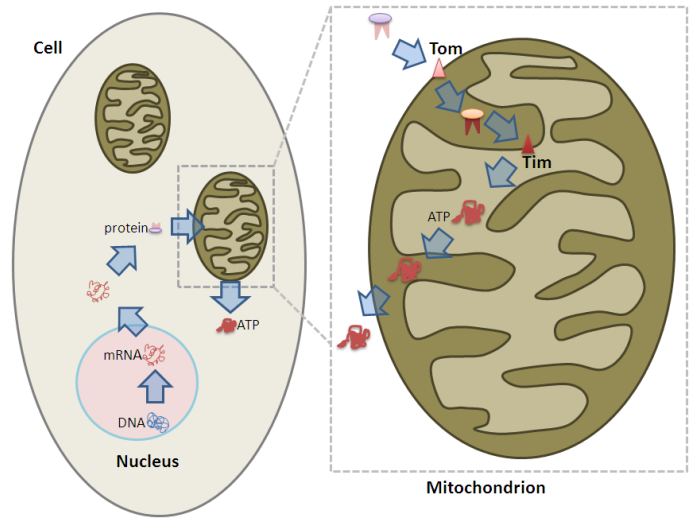}
\end{center}
\caption{\label{FIG_mit}From the cell nucleus to mitochondria.}
\end{minipage}
\end{center}
\end{figure}

The vast majority of proteins destined for the mitochondria are
encoded in the nucleus of the cell and synthesised in the cytoplasm.
These are tagged by a signal sequence, which is recognised by a
receptor protein in the Transporter Outer Membrane complex (TOM).
The signal sequence and adjacent portions of the polypeptide chain
are inserted in the intermembranous space through the TOM complex,
then in the mitochondrion internal space through a Transporter Inner
Membrane complex (TIM). According to this description, we model the
following final steps: (4) protein is detected by TOM and brought
within the intermembranous space, (5) then, through TIM, in the
mitochondrion's inside, (6) where it catalyses the production of
ATP, \AT{(7) that exits the inner, (8) and the outer mitochondrial membranes towards the cell's cytoplasm.}

\section{The calculus}\label{sect:CLSLR}
In this Section we present the Calculus of Looping Sequences with Local Rules (CLSLR).
\subsection{Syntax of CLSLR}
We assume a possibly
infinite alphabet $\EE$ of symbols ranged over by $a,b,c,\ldots$, a set of element variables
$\XX$ ranged over by $x,y,z,\ldots$, a set of sequence variables $\SV$
ranged over by $\xx,\yy,\zz,\ldots$, and a set of term variables
$\TV$ ranged over by $X,Y,Z,\ldots$. All these sets are possibly
infinite and pairwise disjoint. We denote by $\VV$ the set of all
variables, $\VV = \XX \cup \SV \cup \TV $, and with $\genVar$ a generic
variable of $\VV$. Hence, a pattern is a term that may include
variables.

 \begin{definition}\label{d:pattern}[Patterns] \emph{Patterns} $P$, \emph{sequence patterns}
$SP$ and \emph{\LB{local} rules} $R$ of {\em CLS} are given by the
following grammar:
\[
\begin{array}{lll}
P\; & \qqop{::=} &\Loop{SP} \into P \agr P \pipe P\agr \Lr \\
SP\; & \qqop{::=} &\epsilon \agr a  \agr x\agr SP \cdot SP \agr \xx \\
\Lr\; & \qqop{::=} & X  \agr SP \agr \Lr \pipe \Lr \agr R\\
R\; & \qqop{::=} &\Lr \! \mapsto \! \Lr \agr \up\Lr {SP} \! \mapsto \! \up\Lr {SP} \agr \down\Lr {SP} \! \mapsto \! \down\Lr {SP}
\end{array}
\]
where $a$ is a generic element of $\EE$, and $X,\xx$ and $x$ are
generic elements of $\TV,\SV$ and $\XX$, respectively. \AT{Sequence
patterns $SP$ defines patterns for sequences of elements,
$\Loop{SP}$ denotes a closed (looping) sequence which may contain
other patterns through the $\into$ operator, $\pipe$ is used to
denote the parallel composition (juxtaposition) of patterns, and $R$
denotes the syntax of local rules that may either exit ($\up\Lr
{SP}$) or enter ($\down\Lr {SP}$) a closed sequence $SP$. We denote
with
$\PP$ the infinite set of patterns.}\\
A \LB{local} rule {\em $\Lr_1 \! \mapsto \! \Lr_2$ is well formed} if
$\Lr_1 \not\equiv \epsilon$ and $Var(\Lr_2) \subseteq Var(\Lr_1)$,
where $Var(P)$ denotes
set of variables appearing in $P$.\\
 A \LB{local} rule {\em $\up{\Lr_1} {SP_1} \! \mapsto \!
\up{\Lr_2} {SP_2}$ or $\down{\Lr_1} {SP_1} \! \mapsto \!
\down{\Lr_2} {SP_2}$ is well formed} if $\Lr_1 \not\equiv \epsilon$,
$Var(\Lr_2) \subseteq Var(\Lr_1)$ and $Var(SP_2) \subseteq
Var(SP_1)$.
\end{definition}

Terms are patterns containing variables only inside local rules.
Sequences are closed sequence patterns. We denote with $\TT$ the
infinite set of terms, ranged over by $T$, and with $\SSq$ the infinite set of sequences, ranged over by $S$.

 An \emph{instantiation} is a partial function 
 \LB{$\sigma : (\TV \rightarrow \TT) \cup (\SV \rightarrow \SSq) \cup (\XX \rightarrow \EE)$}.
Given $P \in \PP$, with $P \sigma$ we denote the
term obtained by replacing each occurrence of each variable $\genVar
\in \VV$ appearing in $P$ with the corresponding term
$\sigma(\genVar)$, but for local rules, which are left unchanged
by instantiations, i.e., $R \sigma=R$ for all $R$ and $\sigma$. With
$\Sigma$ we denote the set of all the possible instantiations.

\begin{definition}[Structural Congruence] The structural
congruence relations $\equiv_S$ and $\equiv_R$ and $\equiv_P$ are the least
congruence relations on sequence patterns, local rules and on patterns, respectively,
satisfying the rules shown in Figure \ref{sr}.
\end{definition}
\begin{figure}
\[
\begin{array}{c}
SP_1 \cdot ( SP_2 \cdot SP_3 )
    \equiv_S ( SP_1 \cdot SP_2 ) \cdot SP_3 \qquad
SP \cdot \epsilon \equiv_S \epsilon \cdot SP
    \equiv_S SP
\\
SP_1 \equiv_S SP_2\ \mbox{ implies }\ SP_1 \equiv_P SP_2\ \mbox{ and }\
    \Loop{SP_1} \into P \equiv_P \Loop{SP_2} \into P\\
    \Lr_1 \equiv_P \Lr'_1\ \mbox{ and }\Lr_2 \equiv_P \Lr'_2\ \mbox{ imply } \Lr_1 \! \mapsto \! \Lr_2\equiv_R \Lr'_1 \! \mapsto \! \Lr'_2\\
    \Lr_1 \equiv_P {\Lr'}_1\ \mbox{ and }\Lr_2 \equiv_P {\Lr'}_2\ \mbox{ and }\ \ SP_1 \equiv_P {SP'}_1\ \mbox{ and }\ \ SP_2 \equiv_P {SP'}_2\\
    \mbox{ imply } \up{\Lr_1}{SP_1} \! \mapsto \! \up{\Lr_2}{SP_2} \equiv_R \up{{\Lr'}_1}{{SP'}_1} \! \mapsto \! \up{{\Lr'}_2}{{SP'}_2}\ \mbox{ and }\ \down{\Lr_1}{SP_1} \! \mapsto \! \down{\Lr_2}{SP_2} \equiv_R \down{{\Lr'}_1}{{SP'}_1} \! \mapsto \! \down{{\Lr'}_2}{{SP'}_2}\\
    R_1 \equiv_R R_2\ \mbox{ implies }\ R_1 \equiv_P R_2\\
P_1 \pipe P_2 \equiv_P P_2 \pipe P_1 \qquad P_1 \pipe ( P_2 \pipe
P_3 ) \equiv_P (P_1 \pipe P_2) \pipe P_3 \qquad
P \pipe \epsilon \equiv_P P \\
\Loop{\epsilon} \into \epsilon \equiv_P \epsilon \qquad \Loop{SP_1
\cdot SP_2} \into P \equiv_P \Loop{SP_2 \cdot SP_1} \into P
\end{array}
\]
\caption{Structural Congruence}\label{sr}
\end{figure}
Structural congruence rules the  state the associativity of
$\cdot$ and $\pipe$, the commutativity of the latter and the neutral
role of $\epsilon$. Moreover, axiom $\Loop{SP_1
\cdot SP_2} \into P \equiv_P \Loop{SP_2 \cdot SP_1} \into P$ says that looping sequences
can rotate. In the following, for simplicity, we will use $\equiv$
in place of $\equiv_P$.

\subsection{(Parallel) Operational Semantics}

In order to define a reduction step in which (possibly) more than
one rule is applied, following \cite{busi}, we first define the
application of a single rule, either global or local. We resort to
the standard notion of evaluation contexts.
\begin{definition}[Contexts] \label{context} \emph{Evaluation Contexts} $E $ are defined as:
\[
E  ::= \phole \agr E  \pipe T \agr T \pipe E  \agr \Loop{S} \into E
\]
where $T \in \TT$ and $S \in \SSq$. The context $\phole$ is called
the \emph{empty context}. We denote with $\CC$ the infinite set of evaluation
contexts.
\end{definition}
By definition, every evaluation context contains a single hole
$\phole$. Let us assume $E\in \CC$, with $E[T]$ we denote the term
obtained by replacing $\phole$ with $T$ in $E$. The structural
equivalence is extended to contexts in the natural way (by
considering $\phole$ as a new and unique symbol of the alphabet
$\EE$). Note that the shape of evaluation contexts does not permit
to have holes in sequences. A rewrite rule introducing a parallel
composition on the right hand side (as $a \mapsto b \pipe c$)
applied to an element of a sequence (e.g., $m \mycdot a \mycdot m$)
would result into a syntactically incorrect term (in this case $m
\cdot (b\pipe c) \cdot m$).

To enforce the fact that local and global rule applications can be
done in parallel, we underline those subterms that are produced by
the application of the rule involved. Terms matching the
left-hand-side of a (local or global) rule must not have any
underlined subterm. Underlined terms are only used for bookkeeping
in the definition of rule application.

Let $\frTT$ denote the set of terms in which some subterms can be underlined. The erasing mapping
$\defrost: \frTT \mapsto \TT$ erases all underlining obtaining a term generated by the grammar of Definition \ref{d:pattern}.

We first define the application of global or local rules to terms in $\frTT$ which produces terms in $\frTT$.

{\em Global rules} are of the shape $P_1 \mapsto P_2$. They can be
applied to terms only if they occur in a legal evaluation context.

{\em Local rules} are inside compartments, and can be applied
 only if a term matching the left-hand-side of the rule
 occurs in the same compartment.

Notice that global rules have patterns $P$ on both the left and the right-hand
side of the rules, whereas local rules have the less general $L$
patterns. In particular, $L$ patterns do not contain compartments,
and therefore cannot change the nesting structure of the
compartments of a term.

\begin{definition}[Rule Application]\label{def_Fsem}
Given a finite set of global rules $\RR$, {\em the rule application} $\ltrans{}$
is the least relation closed with respect to $\equiv$
defined by:
\[
\prooftree
\begin{array}{c}
P_1 \mapsto P_2 \in \RR \quad \sigma \in \Sigma \quad P_1\sigma
\not\equiv \epsilon
\end{array}
\justifies E [P_1\sigma]\ltrans{P_1 \mapsto P_2,E ,\sigma}
E [\freeze{P_2\sigma}] \using \lab{GRT}
\endprooftree
\]
\[
\prooftree
\begin{array}{c} \sigma \in \Sigma \quad
 \Lr_1\sigma \not\equiv \epsilon
 \end{array}
 \justifies E [\Lr_1 \mapsto \Lr_2\pipe \Lr_1\sigma\pipe T]
 \ltrans{\Lr_1 \mapsto \Lr_2,E ,\sigma,T}
 E [\Lr_1 \mapsto\Lr_2\pipe \freeze{\Lr_2\sigma}\pipe T]
 \using \lab{LR}
\endprooftree
\]
\[\prooftree
\begin{array}{c}
 \sigma \in \Sigma \quad \Lr_1\sigma \not\equiv\epsilon
\end{array}
\justifies
 E [\Loop{S_1\sigma} \into (\up{\Lr_1}{S_1} \mapsto \up{\Lr_2}{S_2} \pipe \Lr_1\sigma \pipe T)]
 \ltrans{\up{\Lr_1}{S_1} \mapsto \up{\Lr_2}{S_2},E ,\sigma,T}
 E [\freeze{\Lr_2\sigma} \pipe \Loop{\freeze{S_2\sigma}} \into (\up{\Lr_1}{S_1} \mapsto \up{\Lr_2}{S_2} \pipe T)]
 \using \using \lab{LR-Out}
\endprooftree
\]
\[\prooftree
\begin{array}{c}
 \sigma \in \Sigma \quad \Lr_1\sigma \not\equiv \epsilon
\end{array}
 \justifies
 E [\down{\Lr_1}{S_1} \mapsto \down{\Lr_2}{S_2} \pipe \Lr_1\sigma \pipe \Loop{S_1\sigma} \into T]
 \ltrans{\down{\Lr_1}{S_1} \mapsto \down{\Lr_2}{S_2},E ,\sigma,T}
 E [\down{\Lr_1}{S_1} \mapsto \down{\Lr_2}{S_2} \pipe \Loop{\freeze{S_2\sigma}} \into (T \pipe \freeze{\Lr_2\sigma})]
 \using \lab{LR-In}
\endprooftree
\]
\end{definition}

With rule $\lab{LR-Out}$ a term or a rule $\Lr_1\sigma$ exits
from the membrane in which it is contained only if the membrane has
as loop the sequence $S_1\sigma$: when outside, $\Lr_1\sigma$ is
transformed into $\Lr_2\sigma$, and the sequence $S_1\sigma$ into
$S_2\sigma$. $\lab{LR-In}$ is similar, but it moves terms or rules
into local membranes. Local rules do not permit to move, create or
delete membranes: only global rules can do that.

Observe that local rules can be dynamically added and deleted both by global and local rules.
A global rule which adds the local rule $R$ is of the shape $P \mapsto R$ and
a global rule which erases the same local rule is of the shape $R \mapsto P$.
A local rule which adds the local rule $R$ is of the shape $\Lr \mapsto R$ and
a local rule which erases the same local rule is of the shape $R \mapsto \Lr$.
Moreover local rules can add local rules in compartments separated by just one membrane, since they can be of the shapes $\up{\Lr}{S} \mapsto \up{R}{S'}$ or $\down{\Lr}{S} \mapsto \down{R}{S'}$.

A reduction step of the parallel semantics $\Ltrans{}$, starting
from a term in $\TT$ applies any number of global or local rules
that could be performed in parallel, producing a final term in $\TT$ (with no underlined subterms).

\begin{definition}[Parallel Reduction]\label{def_Psem} The reduction
$\Ltrans{}$ between terms in $\TT$ is defined by:
\[
\prooftree
\begin{array}{c}
 T=T_0 \ltrans{\labl_0}
 T_1\ltrans{\labl_1}\cdots\ltrans{\labl_n}T_{n+1}\quad n\ge 0\quad T'=\defrost(T_{n+1})
\end{array}
 \justifies T \Ltrans{} T'
\endprooftree
\]
\end{definition}

To justify the definition
of the reduction $\Ltrans{}$ we have to show that the order in which
the local or global rules are applied is not important. The notion of multi-hole context, i.e., of term where some disjoint subterms are replaced by holes is handy.
More precisely the syntax of {\em multi-hole contexts} is:
\[
C  ::= \phole \agr T \agr C  \pipe C \agr \Loop{\phole} \into C  \agr \Loop{S} \into C
\]

\MD{We can show that in a parallel reduction only disjoint subterms can change. Rules \lab{GRT} and \lab{LR} modify just one subterm. Rules \lab{LR-Out} and \lab{LR-In} modify three subterms, i.e., a membrane and a term exiting or entering the membrane, using $\epsilon$ for the missing term.}

\begin{theorem}\label{tpr} If $T \Ltrans{} T'$, then there is a multi-hole context $C[\;]\ldots[\;]$ and terms $T_1,\ldots,T_n$,  $T'_1,\ldots,T'_n$ such that $T\equiv C[T_1]\ldots[T_n]$, $T'\equiv C[T'_1]\ldots[T'_n]$ and for all $1\leq i\leq n$:
\begin{itemize}\item either $C[T^*_1]\ldots[T_i]\ldots[T^*_n]\ltrans{}C[T^*_1]\ldots[T'_i]\ldots[T^*_n]$
\item or there are $i_1,i_2,i_3$ such that $i\in\set{i_1,i_2,i_3}$ and\\
$C[T^*_1]\ldots[T_{i_1}]\pipe\Loop{[T_{i_2}]}\into[T_{i_3}]\ldots[T^*_n]\ltrans{}C[T^*_1]\ldots[T'_{i_1}]\pipe\Loop{[T'_{i_2}]}\into[T'_{i_3}]\ldots[T^*_n]$
\end{itemize}
where $T^*_j$ can be either  $T_j$ (subterm of $T$) or $T'_j$
(subterm of $T'$).
\end{theorem}

\begin{proof}
If $T \Ltrans{} T'$, then for some $U_0,\ldots,U_{m+1}$ we get $T=U_0 \ltrans{}\cdots\ltrans{}U_{m+1}$ where $m\ge 0$ and  $T'=\defrost(U_{m+1}).$
We show by induction on $h\leq m$ and by cases on the last applied reduction rule that {\em
\begin{itemize}\item either $U_h=C[T^*_1]\ldots[T_i]\ldots[T^*_n]$ and $U_{h+1}=C[T^*_1]\ldots[T'_i]\ldots[T^*_n]$
\item or there are $i_1,i_2,i_3$ such that $i\in\set{i_1,i_2,i_3}$ and\\
$U_h=C[T^*_1]\ldots[T_{i_1}]\pipe\Loop{[T_{i_2}]}\into[T_{i_3}]\ldots[T^*_n]$ and $U_{h+1}=C[T^*_1]\ldots[T'_{i_1}]\pipe\Loop{[T'_{i_2}]}\into[T'_{i_3}]\ldots[T^*_n]$
\end{itemize}
where `` $^*$'' can be either ``$\;$'' or `` $'$'' and all terms with $'$ are either underlined or $\epsilon$.}

If the last applied rule is
\[
\prooftree
\begin{array}{c} \sigma \in \Sigma \quad
 \Lr_1\sigma \not\equiv \epsilon
 \end{array}
 \justifies E [\Lr_1 \mapsto \Lr_2\pipe \Lr_1\sigma\pipe V]
 \ltrans{\Lr_1 \mapsto \Lr_2,E ,\sigma,T}
 E [\Lr_1 \mapsto\Lr_2\pipe \freeze{\Lr_2\sigma}\pipe V]
 \using
\endprooftree
\]
then $U_h=E [\Lr_1 \mapsto \Lr_2\pipe \Lr_1\sigma\pipe V]$ and $U_{h+1}=E [\Lr_1 \mapsto\Lr_2\pipe \freeze{\Lr_2\sigma}\pipe V]$. By induction $U_h=C[T^*_1]\ldots[T^*_n]$. Since
$\Lr_1\sigma$ is a subterm of $T$ and $\freeze{\Lr_2\sigma}$ is a subterm of $U_{m+1}$ there must be an index $i$ such that $T_i=\Lr_1\sigma$ and $T'_i=\freeze{\Lr_2\sigma}$.

If the last applied rule is
\[
\prooftree
\begin{array}{c}
 \sigma \in \Sigma \quad \Lr\sigma \not\equiv \epsilon \quad
 \Lr_1\sigma\in\TT\quad S_1\sigma\in\TT
\end{array}
 \justifies
 E [\down{\Lr_1}{S_1} \mapsto \down{\Lr_2}{S_2} \pipe \Lr_1\sigma \pipe \Loop{S_1\sigma} \into V]
 \ltrans{\down{\Lr_1}{S_1} \mapsto \down{\Lr_2}{S_2},E ,\sigma,T}
 E [\down{\Lr_1}{S_1} \mapsto \down{\Lr_2}{S_2} \pipe \Loop{\freeze{S_2\sigma}} \into (V \pipe \freeze{\Lr_2\sigma})]
 \using
\endprooftree
\]
then $U_h=E [\down{\Lr_1}{S_1} \mapsto \down{\Lr_2}{S_2} \pipe \Lr_1\sigma \pipe \Loop{S_1\sigma} \into V]$ and $U_{h+1}=E [\down{\Lr_1}{S_1} \mapsto \down{\Lr_2}{S_2} \pipe \Loop{\freeze{S_2\sigma}} \into (V \pipe \freeze{\Lr_2\sigma})]$. By induction $U_h=C[T^*_1]\ldots[T^*_n]$. Notice that
$\Lr_1\sigma, S_1\sigma$ are subterms of $T$ and $\freeze{\Lr_2\sigma}, \freeze{S_2\sigma}$ are subterm of $U_{m+1}$. Moreover $\Lr_1\sigma$, $S_1\sigma$ and $\epsilon$ in $T$ are replaced by  $\epsilon$, $\freeze{S_2\sigma}$ and $\freeze{\Lr_2\sigma}$ in $U_{m+1}$, respectively.
Therefore there must be indexes $i_1,i_2,i_3$ such that $T_{i_1}=\Lr_1\sigma$, $T'_{i_1}=\epsilon$, $T_{i_2}=S_1\sigma$, $T'_{i_2}=\freeze{S_2\sigma}$, $T_{i_3}=\epsilon$, $T'_{i_3}=\freeze{\Lr_2\sigma}$.

\end{proof}

\begin{example}\label{EX_SYNT_MIT}
[Mitochondria Running Example: Syntax and Reductions] A CLSLR term
representing the mitochondria evolution inside the cell's activity
discussed in the introduction could be:
\[
\begin{array}{ll}
\texttt{CELL} = \Loop{cell}\into( \;      & \texttt{NUCLEUS}\pipe \texttt{MITOCH} \pipe \ldots \pipe \texttt{MITOCH} \pipe\\
                                        & mRNA\! \mapsto \! protein  \pipe  \\
                                        & \down {protein} {Tom} \! \mapsto \! \down {protein} {Tom} \quad)
\end{array}
\]

\AT{A cell is composed by its membrane (here just represented by the element $cell$) and its content (in this case, the nucleus, a certain number of mitochondria and a few rules modelling the activity of interest). In particular, the two rules above model the steps (3) and (4), respectively, of the example schematised in the introduction.}

\AT{Assuming that \DNA\ is the sequence of genes representing
the cell's DNA, and \GENE\ is the particular gene (contained in \DNA ) codifying the protein, we define the nucleus of the cell with the CLSLR term:}
\[\begin{array}{ll}
\texttt{NUCLEUS}=\Loop{nucleus}\into( \; & \DNA \pipe \\
                                         & \xx\cdot\GENE\cdot\yy \! \mapsto \! (\xx\cdot\GENE\cdot\yy\pipe mRNA)\pipe \\
                                         & \up{mRNA}{nucleus}\! \mapsto \! \up{mRNA}{nucleus} \quad)
\end{array}\]
\AT{Note that the first of the two rules above models step (1) of our example (DNA transcription into mRNA), the second one (mRNA exits the nucleus) models step (2).}

\AT{The mitochondria of our model are composed of a membrane, on which we point out the $Tom$ complex, containing an inner membrane ($\texttt{INN\_MITOCH}$) and a couple of rules:}
\[\begin{array}{ll}
\texttt{MITOCH}= \Loop{Tom}\into( \; & \texttt{INN\_MITOCH}\pipe \\
    & \down {protein} {Tim} \! \mapsto \! \down {(Mit_A\! \mapsto \! (Mit_A\pipe \textit{ATP}))} {Tim}\pipe \\
    & \up{\textit{ATP}}{\xx}\! \mapsto \! \up{\textit{ATP}}{\xx} \quad)\\
\end{array}\]
\AT{where we denote with the element $Mit_A$ a mitochondrial factor
inside the inner membrane (activated by our protein), necessary to
produce ATP. In particular, the protein, when in the intermembranous
space, is moved through $Tim$ inside the inner mitochondrial space
(step (5) of our example) and then transformed into the newly
generated rule $Mit_A\! \mapsto \! (Mit_A\pipe \textit{ATP})$ which
will lead the production of \textit{ATP} (step (6) of our example).}

\AT{Finally, in \texttt{INN\_MITOCH} we point out the $Tim$
complex:}
\[\begin{array}{ll}
\texttt{INN\_MITOCH}=\Loop{Tim}\into( \; & Mit_A \pipe \\
                                        & \up{\textit{ATP}}{\xx}\! \mapsto \! \up{\textit{ATP}}{\xx} \quad)\\
\end{array}\]
\AT{Both in \texttt{MITOCH} and \texttt{INN\_MITOCH} we have the rules needed to transport the \textit{ATP} towards the cell's cytoplasm (steps (7) and (8) of the example).}

A possible \AT{(parallel)} reduction of this term, when \GENE\ initially produces
a certain number of $mRNA$  is (by focusing only on the more interesting
changes) shown in Figure \ref{red}.
\begin{figure}
\[\begin{array}{lll}
\ldots&\Ltrans{}^+&\Loop{cell}\into(\Loop{nucleus}\into(mRNA\pipe\ldots \pipe mRNA\pipe\ldots)\pipe\ldots)\\
&\Ltrans{}^+&\Loop{cell}\into(mRNA\pipe\ldots \pipe mRNA\pipe\ldots)\\
&\Ltrans{}&\Loop{cell}\into(protein\pipe\ldots \pipe protein\pipe\ldots)\\
&\Ltrans{}&\Loop{cell}\into(\Loop{Tom}\into(protein\pipe\ldots)\pipe\ldots\pipe\Loop{Tom}\into(protein\pipe\ldots)\pipe\ldots)\\
&\Ltrans{}&\Loop{cell}\into(\Loop{Tom}\into(\Loop{Tim}\into(Mit_A\! \mapsto \! (Mit_A\pipe \textit{ATP})\pipe\ldots)\pipe\ldots)\pipe\\
&&\qquad\qquad\Loop{Tom}\into(\Loop{Tim}\into(Mit_A\! \mapsto \! (Mit_A\pipe \textit{ATP})\pipe\ldots)\pipe\ldots)\pipe\ldots)\\
&\Ltrans{}&\Loop{cell}\into(\Loop{Tom}\into(\Loop{Tim}\into(\textit{ATP}\pipe\ldots)\pipe\ldots)\pipe\\
&&\qquad\qquad\Loop{Tom}\into(\Loop{Tim}\into(\textit{ATP}\pipe\ldots)\pipe\ldots)\pipe\ldots)\\
&\Ltrans{}&\Loop{cell}\into(\Loop{Tom}\into(\textit{ATP}\pipe\Loop{Tim}\into(\textit{ATP}\pipe\ldots)\pipe\ldots)\pipe\\
&&\qquad\qquad\Loop{Tom}\into(\textit{ATP}\pipe\Loop{Tim}\into(\textit{ATP}\pipe\ldots)\pipe\ldots)\pipe\ldots)\\
&\Ltrans{}&\Loop{cell}\into(\textit{ATP}\pipe\ldots \pipe \textit{ATP}\pipe\Loop{Tom}\into(\textit{ATP}\pipe\Loop{Tim}\into(\textit{ATP}\pipe\ldots)\pipe\ldots)\pipe\\
&&\qquad\qquad\Loop{Tom}\into(\textit{ATP}\pipe\Loop{Tim}\into(\textit{ATP}\pipe\ldots)\pipe\ldots)\pipe\ldots)\\
\end{array}\]
\caption{Mitochondria evolution}\label{red}
\end{figure}
The $\textit{ATP}$ produced in the last but two reductions  in
$\texttt{INN\_MITOCH}$ moves to $\texttt{MITOCH}$ in the last but one reduction
and new $\textit{ATP}$ is produced in $\texttt{INN\_MITOCH}$. In the last
reduction, the firstly generated $\textit{ATP}$ moves to the cell,
the secondly generated $\textit{ATP}$ moves to $\texttt{MITOCH}$ and
new $\textit{ATP}$ is produced in $\texttt{INN\_MITOCH}$.

\end{example}

\section{Types}

In this section we introduce a type system that enforces the fact
that compartments must contain rules having specific features. E.g.,
in \cite{OP11} the following {\em rule features} for $\Lr_1 \mapsto
\Lr_2$ are suggested:
\begin{itemize}
\item the rule is \emph{deleting} if $Vars(\Lr_1)\supset Vars( \Lr_2)$ (denoted by \df);
\item the rule is \emph{replicating} if some variable in $ \Lr_2$ occurs twice (denoted by \rf);
\item the rule is \emph{splitting} if $ \Lr_1$ has a subterm containing two different variables (denoted by \s);
\item the rule is \emph{equating} if some variable in $ \Lr_1$ occurs twice (denoted by
\ef).
\end{itemize}
This kind of features reflects a structure of rewrite features which
could be common for rewrite systems in general. Other,
model-dependent, features could be defined to reflect peculiarities
and properties of the particular model under investigation. The
features of the rules allowed in a compartment are controlled by the
wrapping sequence of the compartment. Our {\em typing system} and
the consequent {\em typed parallel reduction} ensure that, in spite
of the facts that reducing a term may move rules in and out of
compartments, compartments always contain rules permitted by their
wrapping sequence. In addition to the previous features of rules we
say that:
\begin{itemize}
\item the feature of rule $\up{\Lr_1}{S_1} \mapsto \up{\Lr_2}{S_2}$
is that it is an \emph{out rule} (denoted by \of);
\item the feature of rule $\down{\Lr_1}{S_1} \mapsto \down{\Lr_2}{S_2}$
is that it is an \emph{in rule} (denoted by \fif).
\end{itemize}
To express the control of the wrapping sequence over the content of
the compartment, we associate a subset
$\varphi$ of $\set{\df,\rf,\s,\ef,\of,\fif}$ to every element in $\EE$. This is called a {\em
membrane type}. We use $\La$ to denote a classification of elements.
The type assignment in Figure \ref{fig_TyMR}, where a basis $\T$
assigns membrane types to element and sequence variables, defines
the type of a sequence as the union of the membrane types of its
elements.

\begin{figure}[h!]
\begin{center}
\framebox{ $
\begin{array}{c}
\begin{array}{c@{\quad\quad\quad}c@{\quad\quad\quad}c}
\TypeSeq\T\epsilon{\emptyset}\quad\lab{TSeps} & \TypeSeq{
\T,\GenericVariable:\varphi}{\GenericVariable}{\varphi}
\quad\lab{TSvar}&
 \prooftree
a:\varphi\in \La \justifies \TypeSeq{ \T}a \varphi \using
\lab{TSelm}
\endprooftree
\end{array}
\\ \\
\prooftree \TypeSeq{ \T}{SP}{\varphi}\quad
\TypeSeq{\T}{SP'}{\varphi'} \justifies \TypeSeq{\T}{SP\cdot
SP'}{\varphi \LB{\cup} \varphi'} \using \lab{TSseq}
\endprooftree
\end{array}
$ }
\end{center}
\caption{Typing Rules for Membranes} \label{fig_TyMR}
\end{figure}
To define the type of patterns, that may contain parallel
(composition) of rules, we consider:
\begin{enumerate}
\item  the features of the rules, contained in the pattern, and
\item in case there are
output rules the type of the rules that are emitted by these output
rules.
\end{enumerate}
Therefore a {\em \pat\ type}, denoted by $\tau$, is a sequence of
membrane types, i.e., $\tau\in\set{\varphi}^{\ast}$. With
$\emptyseq$ we denote the empty sequence. If the parallel
composition of local rules $R_1\pipe\cdots\pipe R_n$, $n\geq 0$, has type
$\conc{\varphi}{\tau}$, then $\varphi$ is the union of the features
of the rules $R_i$ ($1\leq i\leq n$), and $\tau$ is the type of the
parallel composition of rules in $\Lr'$ for those $\Lr'$ such that
$R_i=\upRule{\Lr}{SP}{\Lr'}{SP'}$ for some $i$, $1\leq i\leq n$ (the
type of the parallel composition of the rules that are emitted). If
no rule is emitted, then $\tau=\emptyseq$.

{\em Union} of \pat\ types, $\Rtunion$, is defined inductively by:
\begin{itemize}
  \item $\emptyseq\Rtunion\tau=\tau\Rtunion\emptyseq=\tau$, and
   \item
   $\conc{\varphi_1}{\tau_1}\Rtunion\conc{\varphi_2}{\tau_2}=\conc{(\varphi_1\cup\varphi_2)}{(\tau_1\Rtunion\tau_2)}$.
  \end{itemize}
and {\em containment}, $\Rtsub$, is defined by:
\begin{itemize}
 \item $\emptyseq\Rtsub\tau$
 \item $\conc{\varphi}{\tau}\Rtsub\conc{\varphi'}{\tau'}$ if
  $\varphi\subseteq\varphi'$ and $\tau\Rtsub\tau'$.
 \end{itemize}
The judgment $\TypePat{ \T}{P}{\tau}$, defined in Figure
\ref{fig_TyR}, asserts that the pattern $P$ is {\em well formed} and
has \pat\ type $\tau$, assuming the basis $\T$, which assigns membrane
types to element  and sequence variables and pattern types to term
variables. The judgement $\TypeGR{ \T} {P_1 \mapsto P_2} {\Phi_2}$ in last rule defines well-formedness of
global rules.
\begin{figure}[h!]
\begin{center}
\framebox{
$
\begin{array}{c}
{ \TypePat{ \T
\LB{,X:\tPat}}{X}{\LB{\tPat}}\quad\lab{Tvar}}\quad\quad\quad {
\TypePat{ \T}{SP}{\emptyseq}\quad\lab{Tseq} }
\\ \\
\prooftree\TypePat\T{L_2}{\tau}
\justifies \TypePat{ \T} {\Lr_1 \mapsto \Lr_2} {{\feature(\Lr_1
\mapsto \Lr_2)}\Rtunion\tau} \using \lab{TRloc}
\endprooftree
\\ \\
\prooftree
 \begin{array}{l}
\TypePat\T{L_2}{\tau}\quad\TypeSeq\T{S_1}{\varphi_1}\quad\TypeSeq\T{S_2}{\varphi_2}\quad\varphi_1\LB{\subseteq}\varphi_2
 \end{array}
\justifies \TypePat{ \T} {\up{\Lr_1}{S_1} \mapsto \up{\Lr_2}{S_2}} {\conc {\set\of}{\tau}}
\using \lab{TRlocOut}
\endprooftree \\ \\
 \prooftree
\TypePat\T{L_2}{{\varphi}\LB{\conc{\
}{\tau}}}\quad\TypeSeq\T{S_1}{\varphi_1}\quad\TypeSeq\T{S_2}{\varphi_2}\quad\varphi\LB{\cup}\varphi_1\LB{\subseteq}\varphi_2\justifies
\TypePat{ \T} {\down{\Lr_1}{S_1} \mapsto \down{\Lr_2}{S_2}}
{{\set\fif\LB{\Rtunion \tau}}} \using \lab{TRlocIn}
\endprooftree \\ \\
\prooftree \TypePat{ \T}{P}{\tau}\quad \TypePat{\T}{P'}{\tau'}
\justifies \TypePat{\T}{P\pipe P'}{\tau \Rtunion \tau'}
\using \lab{Tpar}
\endprooftree \\ \\
\prooftree
\begin{array}{l}
\TypeSeq{ \T}{SP}{\varphi}\quad
\TypePat{\T}{P}{\conc{\varphi'}{\tau'}}\quad \varphi' \LB{\subseteq} \varphi
\end{array}
\justifies \TypePat{ \T}{\Loop{SP} \into P}{\tau'}
\using \lab{Tcomp}
\endprooftree \\ \\
\prooftree
\TypePat{\T } {P_1}{\tau_1}\quad\TypePat\T{P_2}{\tau_2}\quad\tau_2\Rtsub\tau_1
\justifies \TypeGR{ \T} {P_1 \mapsto P_2} {\Phi_2}
\using \lab{TRglob}
\endprooftree
\end{array}
$
}
\end{center}
\caption{Typing Rules for Patterns and Global Rules}
\label{fig_TyR}
\end{figure}
It is easy to verify that the typing rules in Figures \ref{fig_TyMR}
and \ref{fig_TyR} enjoy weakening, i.e., if $\T \subseteq \T'$ then
$\TypeSeq{\T}{SP}{\tSeq}$ implies $\TypeSeq{\T'}{SP}{\tSeq}$,
$\TypePat{\T}{P}{\tPat}$ implies $\TypePat{\T'}{P}{\tPat}$, and $\TypeGR{ \T} {P_1 \mapsto P_2} {\Phi_2}$ implies $\TypeGR{ \T'} {P_1 \mapsto P_2} {\Phi_2}$.

Rule \lab{Tvar} asserts that a term variable is well typed when its
\pat\ type is found in the basis. Rule \lab{Tseq} asserts that,
since a sequence does not contain rules, its \pat\ type is empty.
Rule \lab{TRloc} asserts that the type of a local rule $R=\Lr_1
\mapsto \Lr_2$ is the union of the set of features of the rule $R$,
denoted by $\feature(R)$, and the  \pat\ type of its right-hand-side
$\Lr_2$. This is because once the rule is applied an instance of the
pattern $\Lr_2$ will substitute the instance of its left-hand-side
$\Lr_1$. Rule \lab{TRlocOut} checks that the features of rules
permitted by the membrane $S_2$ include the one permitted by $S_1$,
so that if the compartment was well formed before applying
$\up{\Lr_1}{S_1} \mapsto \up{\Lr_2}{S_2}$, it will be well formed
afterwards (when $S_2$ replaces $S_1$). Moreover, the \pat\ type of
the rule is $\set\of$, concatenated with the \pat\ type of $\Lr_2$,
since $\Lr_2$ is the pattern sent outside the compartment. Rule
\lab{TRlocIn} checks rule $\down{\Lr_1}{S_1} \mapsto
\down{\Lr_2}{S_2}$. Since the pattern $\Lr_2$ will get into a
compartment with membrane $S_1$, the membrane $S_2$, that replaces
$S_1$, must permit all the features of rules that were permitted
before, and moreover, it permits the features of the rules in
$\Lr_2$. The type is $\set\fif$ union the type of the patterns that
are emitted by the out rules contained in $\Lr_2$. Rule \lab{Tpar}
enforces the fact that the patterns in parallel are both well formed
and the final \pat\ type is the union of the two \pat\ types. Rule
\lab{Tcomp} checks that a compartment contain only rules whose
features are permitted by its wrapping sequence. The \pat\ type of
the compartment is the type of the rules that are emitted. Finally,
rule \lab{TRGlob} says that
the global rule $P_1 \mapsto P_2$ is well formed in case the pattern $P_2$ that will
replace $P_1$ has less features, so that it is permitted by all the
compartments in which $P_1$ is permitted.

As we can see from rule \lab{TRLoc} the type system is independent
from the specific set of features considered. Any syntactic
characterisation of rules could be considered for a feature.

Based on the previous typing system we define a {\em typed
semantics}, that preserves well-formedness of terms. Let an
instantiation $\sigma$ {\em agree} with a basis $\T$ (notation
$\sigma\in\Sigma_\T$) if $x:\tSeq\in \T$ implies
$\TypeSeq{}{\sigma(x)}{\tSeq}$, $\xx:\tSeq\in \T$ implies
$\TypeSeq{}{\sigma(\xx)}{\tSeq}$, and $X:\tPat\in \T$ implies
$\TypePat{}{\sigma(X)}{\tPat}$. This is sound since the judgments
$\vdash_s$ use assumptions on element and sequence variables, while
the the judgments $\vdash_p$ use assumptions on term variables.

\begin{definition}[Typed Rule Application]\label{def_Tsem}
Given a finite set of global rules $\RR$, the {\em typed rule application} $\ttrans{}$ is the least relation closed with respect to $\equiv$
and satisfying the following rules:
\[
\prooftree
\begin{array}{c}
\RR_\T = \{P_1 \mapsto P_2 \in \RR \pipe \TypeGR\T{P_1 \mapsto P_2}{}\}\\
P_1 \mapsto P_2 \in \RR_\T \quad \sigma \in \Sigma_\T \quad
P_1\sigma \not\equiv \epsilon
\end{array}
\justifies E[P_1\sigma] \ttrans{} E[\freeze{P_2\sigma}]
\using \lab{T-GRT}
\endprooftree
\]
\[\prooftree
\begin{array}{c} \sigma \in \Sigma_\T \quad \Lr_1\sigma \not\equiv
\epsilon   \end{array} \justifies E[\Lr_1 \mapsto \Lr_2\pipe
\Lr_1\sigma] \ttrans{} E[\Lr_1 \mapsto \Lr_2\pipe \freeze{\Lr_2\sigma}]
\using \lab{T-LR}
\endprooftree
\]
\[\prooftree
\begin{array}{c} \sigma \in \Sigma_\T \quad \Lr_1\sigma \not\equiv
\epsilon \quad \end{array} 
\justifies E[\Loop{S_1\sigma} \into (\up{\Lr_1}{S_1} \mapsto
\up{\Lr_2}{S_2} \pipe \Lr_1\sigma \pipe T)] \ttrans{} E[\freeze{\Lr_2\sigma}
\pipe \Loop{\freeze{S_2\sigma}} \into (\up{\Lr_1}{S_1} \mapsto
\up{\Lr_2}{S_2} \pipe T)] \using \lab{T-LR-Out}
\endprooftree
\]
\[\prooftree
\begin{array}{c} \sigma \in \Sigma_\T \quad \Lr_1\sigma \not\equiv
\epsilon  \end{array} 
\justifies E[\down{\Lr_1}{S_1} \mapsto \down{\Lr_2}{S_2} \pipe
\Lr_1\sigma \pipe \Loop{S_1\sigma} \into T] \ttrans{}
E[\down{\Lr_1}{S_1} \mapsto \down{\Lr_2}{S_2} \pipe \Loop{\freeze{S_2\sigma}}
\into (T \pipe \freeze{\Lr_2\sigma})] \using \lab{T-LR-In}
\endprooftree
\]
\end{definition}
\begin{definition}[Typed Parallel Reduction]\label{def_PsemT} The reduction
$\Ttrans{}$ between term in $\TT$ is defined by:
\[
\prooftree
\begin{array}{c}
 T=T_0 \ttrans{\labl_0}
 T_1\ttrans{\labl_1}\cdots\ttrans{\labl_n}T_{n+1}\quad n\ge 0\quad T'=\defrost(T_{n+1})
\end{array}
 \justifies T \Ttrans{} T'
\endprooftree
\]
\end{definition}

The property enforced by the type system is that well-typed terms
reduce to well-typed terms: the proof is the content of the Appendix.

\smallskip

\begin{theorem}[Subject Reduction]\label{theoPreserveCorrect}
If $\TypePat{}{T}\tPat$ and $T \Ttrans{} T'$, then
$\TypePat{}{T'}{\tPat'}$ for some $\tPat'\Rtsub\tPat$.
\end{theorem}

\smallskip

\begin{example}[Mitochondria Running Example: Typing]
Let the rules used in Example~\ref{EX_SYNT_MIT} be labelled by:
\begin{itemize}
  \item $\Rg=\xx\cdot\GENE\cdot\yy \! \mapsto \! (\xx\cdot\GENE\cdot\yy\pipe mRNA)$,
  \item $\RmR=mRNA\! \mapsto \! protein$,
  \item $\Rpdo=\down {protein} {Tom} \! \mapsto \! \down {protein} {Tom}$,
  \item $\Rmu=\up{mRNA}{nucleus}\! \mapsto \! \up{mRNA}{nucleus}$,
  \item $\Rpdi=\down {protein} {Tim} \! \mapsto \! \down {\Rmic} {Tim}$,
  \item $\Rmic=Mit_A\! \mapsto \! (Mit_A\pipe \textit{ATP})$
  \item $\Rau=\up{\textit{ATP}}{\xx}\! \mapsto \!
  \up{\textit{ATP}}{\xx}$.
 \end{itemize}
 Let $\varphi_g=\feature(\Rg)$, $\varphi_m=\feature(\RmR)$,
$\varphi_a=\feature(\Rmic)$.
The term representing our model can be typed if $\La$ contains appropriate membrane types
for the elements which occur in the membranes, i.e.:
\[\set{cell:\varphi_{cell}, nucleus:\varphi_{nucleus},
Tom:\varphi_{Tom}, Tim:\varphi_{Tim}}\subseteq\La\] where
$\set{\fif}\cup\varphi_m\subseteq\varphi_{cell}$,
$\set{\of}\cup\varphi_g\subseteq\varphi_{nucleus}$,
$\set{\of,\fif}\subseteq\varphi_{Tom}$, and  $\set{\of}\cup\varphi_a\in\varphi_{Tim}$.
In this case the given parallel reduction is also a typed parallel
reduction for this term.

\smallskip

 We can type the $\texttt{MITOCH}$ 
with the following derivations:

\[
\hspace{-1.2cm}
\prooftree
\begin{array}{l}
\prooftree
\begin{array}{l}
\prooftree
\begin{array}{l}
\TypePat{\T}{Mit_A \pipe \textit{ATP}}{\emptyset}
\end{array}
\justifies \TypePat{\T}{\Rmic}{\varphi_a} \using \lab{TRloc}
\endprooftree
\quad \TypeSeq{\T}{Tim}{\varphi_{Tim}} \quad
\varphi_a\subseteq\varphi_{Tim}
\end{array}
\justifies \TypePat{\T}{\Rpdi}{\set\fif}
\using \lab{TRlocIn}
\endprooftree
\quad
\TypePat{\T}{\Rau}{\set\of}
\end{array}
\justifies \TypePat{\T}{\Rpdi \pipe \Rau }{\set{\fif, \of}}
\using \lab{Tpar}
\endprooftree
\]

\[
\prooftree
\begin{array}{l}
\TypeSeq{\T}{Tom}{\varphi_{Tom}}\quad
\prooftree
\begin{array}{l}
\TypePat{\T}{\texttt{INN\_MITOCH}}{\emptyset}\quad
     \TypePat{\T}{\Rpdi \pipe \Rau }{\set{\fif, \of}}
\end{array}
\justifies \TypePat{\T}{\texttt{INN\_MITOCH}\pipe \Rpdi \pipe \Rau}{\set{\fif, \of}}
\using \lab{Tpar}
\endprooftree
\quad \set{\fif, \of} \subseteq \varphi_{Tom}
\end{array}
\justifies \TypePat{\T}{\Loop{Tom}\into(\texttt{INN\_MITOCH}\pipe \Rpdi \pipe \Rau)}{\emptyset}
\using \lab{Tcomp}
\endprooftree
\]

\noindent where we can type $\texttt{INN\_MITOCH}$ with:
\[
\hspace{-1.2cm}
\prooftree
\begin{array}{c}
\prooftree
\begin{array}{l}
\TypePat{\T}{Mit_A}{\emptyset}\quad
\prooftree
\begin{array}{l}
\TypePat{\T}{\textit{ATP}}{\emptyset} \quad
\TypeSeq{\T}{\xx}{\varphi_{Tom}} \quad \varphi_{Tom}\subseteq
\varphi_{Tom}
\end{array}
\justifies \TypePat{\T}{\Rau}{\conc {\set\of}{\emptyset}}
\using \lab{TRlocOut}
\endprooftree
\end{array}
\justifies \TypePat{\T}{Mit_A\pipe\Rau}{\set\of}
\using \lab{Tpar}
\endprooftree\\[2mm]
\TypeSeq{\T}{Tim}{\varphi_{Tim}}\quad\quad \set\of \subseteq \varphi_{Tim}
\end{array}
\justifies \TypePat{\T}{\Loop{Tim}\into(Mit_A\pipe\Rau)}{\emptyset}
\using \lab{Tcomp}
\endprooftree
\]

\end{example}

\section{Related Works and Conclusions}

{\em $\kappa$-calculus} is a formalism proposed by Danos and
Laneve~\cite{DanLan04} that idealises protein-protein interactions
using graphs and graph-rewriting operations. A protein is a node
with a fixed number of sites, that may be bound or free. Proteins
may be assembled into complexes by connecting two-by-two bound sites
of proteins, thus building connected graphs. Collections of proteins
and complexes evolve by means of reactions, which may create or
remove proteins and bounds between proteins: $\kappa$-calculus
essentially deals with complexations and decomplexations, where
complexation is a combination of substances into a new substance
called complex, and the decomplexation is the reaction inverse to
complexation, when a complex is dissociated into smaller parts.
These rules contain variables and are pattern-based, therefore may
be applied in different contexts. Even if $\kappa$-calculus does not
deal with membranes,
we have in common the use of variables and contexts for rule
application. Moreover, both approaches emphasise the key rule of the
surface components, in proteins ($\kappa$-calculus) or membranes
(CLSLR), for biological modelling.

{\em P-Systems} \cite{Pau02} are a biologically inspired
computational model. A P-System is formed by a membrane structure:
each membrane may contain molecules, represented by symbols of an
alphabet, other membranes and rules. The rules contained into a
membrane can be applied only to the symbols contained in the same
membrane: these symbols can be modified or moved across membranes.
The key feature of P-Systems is the maximal parallelism, i.e., in a
single evolution step all symbols in all membranes evolve in
parallel, and every applicable rule is applied as many times as
possible. Locality and intrinsic parallelism of rules are also
present in our approach, but in CLSLR the level of parallelism is
not necessarily maximal, and moreover not only molecules but also
rules can be created, deleted or moved across membranes. In both
approaches the local rules cannot describe some possible biological
behaviours such as fusion, deletion or creation of membranes.
P-Systems are not so flexible in the description of new activities
observed on membranes without extending the formalism to model such
activities.
In CLSLR this limit is overcome by global rules, that contain generic
patterns.

In rewrite system models, the term (describing the systems under
consideration) and the list of rules (describing the system's
evolution) could be considered as separate (written on two different
sheets of paper). In this work, we have presented a calculus with
global (separate from the system) and local (dynamic and system
intrinsic) rewrite rules. While global rules can, as usual, be
applied anywhere in a given term, local rules can only be applied in
the compartment on which they are defined. Local rules are equipped
with dynamic features: they can be created, moved and erased.


As it happens for P-Systems, local rules are intrinsically parallel.
Indeed, expressing rules that are local to well delimited
compartments, and with the possibility to define systems with
multiple, parallel, compartments, naturally leads to the definition
of a parallel semantics.

As a future work, in the lines of~\cite{BDGT12,DGT09b,Bio11}, we plan to investigate how to adapt this model
with a quantitative semantics, also studying the limits and
constraints imposed by a parallel semantics.

\bibliographystyle{eptcs}
\bibliography{biblio}

\appendix
\section{APPENDIX}
\begin{lemma}[Inversion Lemma]\label{lemmaInversion}
\begin{enumerate}
 \item\label{LITSeps} If $\TypeSeq{ \T}\epsilon \tSeq$, then $\tSeq=\emptyseq$.
 \item\label{LITSvar} If $\TypeSeq{ \T}\GenericVariable \tSeq$, then 
  $\chi:\tSeq\in \T$.
 \item\label{LITSelm} If $\TypeSeq{ \T}a \tSeq$, then 
  $a:\tSeq\in \La$.
 \item\label{LITSseq} If $\TypeSeq{\T}{SP\cdot SP'}{\tSeq}$, then 
  $\TypeSeq{ \T}{SP}{\tSeq_1}$, $\TypeSeq{\T}{SP'}{\tSeq_2}$ and $\tSeq = \tSeq_1 \PG{\Rtunion} \tSeq_2$.
 \item\label{LITvar} If $\TypePat{\T}{X}{\tPat}$, then $X:\tPat\in\T$.
 \item\label{LITseq} If $\TypePat{\T}{SP}{\tPat}$, then $\tPat=\emptyseq$.
 \item\label{LITRloc} If $\TypePat{ \T} {\Lr_1 \mapsto \Lr_2} {\tPat}$, then $\tPat={\feature(\Lr_1 \mapsto \Lr_2)\Rtunion\tPat'}$ and  $\TypePat{ \T} {\Lr_2}{\tPat'}$.
 \item\label{LITRlocOut} If $ \TypePat{ \T} {\up{\Lr_1}{S_1} \mapsto \up{\Lr_2}{S_2}} {\tPat}$, then $\tPat=\conc {\set\of}{\tPat'}$, $\TypePat\T{L_2}{\tPat'}$, $\TypeSeq\T{S_1}{\tSeq_1}$, $\TypeSeq\T{S_2}{\tSeq_2}$ and $\tSeq_1\LB{\subseteq}\tSeq_2$.
 \item\label{LITRlocIn} If $\TypePat{ \T} {\down{\Lr_1}{S_1} \mapsto \down{\Lr_2}{S_2}} {\tPat}$, then $\tPat={\set\fif\Rtunion \tPat'}$, $\TypePat\T{L_2}{\conc{\tSeq}{\tPat'}}$, $\TypeSeq\T{S_1}{\tSeq_1}$, $\TypeSeq\T{S_2}{\tSeq_2}$ and $\tSeq \LB{\cup} \tSeq_1\LB{\subseteq}\tSeq_2$.
 \item\label{LITpar} If $\TypePat{\T}{P\pipe P'}{\tPat}$, then 
 $\tPat = \tPat_1 \Rtunion \tPat_2$, $\TypePat{ \T}{P}{\tPat_1}$ and $\TypePat{\T}{P'}{\tPat_2}$.
 \item\label{LITcomp} If $\TypePat{ \T}{\Loop{SP} \into P}{\tPat}$, then $\TypeSeq{ \T}{SP}{\tSeq}$, $\TypePat{\T}{P}{\conc{\tSeq'}{\tPat}}$ and $\tSeq' \LB{\subseteq} \tSeq$.
 \item\label{LITRglob} If $\TypeGR{ \T} {P_1 \mapsto P_2} {\Phi_2}$, then $\TypePat{\T } {P_1}{\tPat_1}$, $\TypePat\T{P_2}{\tPat_2}$ and $\tPat_2\Rtsub\tPat_1$.
\end{enumerate}
\begin{proof}
Immediate from 
the typing rules in Figures
\ref{fig_TyMR} and \ref{fig_TyR}.
\end{proof}
\end{lemma}

\begin{lemma}\label{lemmaHoleTyped}
If $\TypePat{\T}{E[P]}\tPat$ then
\begin{enumerate}
  \item\label{lemmaHoleTyped1} $\TypePat{\T}{P}{\tPat_0}$ for some $\tPat_0$, and
  \item\label{lemmaHoleTyped3} if $P'$ is such that $\TypePat{\T}{P'}{\tPat'}$ with $\tPat'\Rtsub\tPat_0$, then
  $\TypePat{\T}{E[P']}\tPat''$  with $\tPat''\Rtsub\tPat$.
\end{enumerate}
\begin{proof}
By induction on the definition of contexts.
\begin{itemize}
 \item If $E = \phole$, then $E[P] = P$, and so $\TypePat{\T}{P}{\tPat}$. Since in this case $E[P'] = P'$, and $\TypePat{\T}{P'}{\tPat'}$ with $\tPat'\Rtsub\tPat$ by hypothesis, then $\TypePat{\T}{E[P']}\tPat'$.
 \item  If $E = E' \pipe T$, then $E[P] = E'[P] \pipe T$. From Lemma \ref{lemmaInversion}(\ref{LITpar}) we derive $\TypePat{\T}{E'[P]}{\tPat_1}$ and $\TypePat{\T}{T}{\tPat_2}$, with $\tPat_1 \Rtunion \tPat_2 = \tPat$. By induction hypothesis on $E'[P]$ we get $\TypePat{\T}{P}{\tPat_0}$ and $\TypePat{\T}{E'[P']}\tPat'_1$  with $\tPat'_1\Rtsub\tPat_1$. Applying rule $\lab{Tpar}$ we conclude $\TypePat{\T}{E[P']}\tPat''$ with $\tPat'' = \tPat'_1 \Rtunion \tPat_2$, and then $\tPat''\Rtsub\tPat$.
 \item If $E = \Loop{S} \into E'$, then $E[P] = \Loop{S} \into E'[P]$. From Lemma \ref{lemmaInversion}(\ref{LITcomp}) we derive $\TypePat{\T}{S}{\tSeq_0}$, and $\TypePat{\T}{E'[P]}{\conc{\tSeq}{\tPat}}$, with $\tSeq \subseteq \tSeq_0$. By induction hypothesis on $E'[P]$ we get $\TypePat{\T}{P}{\tPat_0}$, and $\TypePat{\T}{E'[P']}{\conc{\tSeq'}{\tPat'}}$ with $\conc{\tSeq'}{\tPat'} \Rtsub \conc{\tSeq}{\tPat}$. Applying rule $\lab{Tcomp}$ we conclude $\TypePat{\T}{E[P']}\tPat'$, with $\tPat' \Rtsub \tPat$.
\end{itemize}
\end{proof}
\end{lemma}

\begin{lemma}\label{LemmaS}
If $\sigma\in\Sigma_\T$, then $\TypeSeq{}{SP\sigma}{\tSeq}$ if and
only if $\TypeSeq{\T}{SP}{\tSeq}$.
\begin{proof}
($\Leftarrow$) By induction on $\TypeSeq{\T}{SP}{\tSeq}$. Consider
the last applied rule.
\begin{itemize}
 \item If the rule is $\lab{TSvar}$, the proof follows from $\sigma\in\Sigma_\T$. For rules $\lab{TSeps}$ and $\lab{TSelm}$, the fact that $SP$ is a term implies that $SP\sigma = SP$, and, moreover, it is typable from the empty environment.
 \item Rule $\lab{TSseq}$. In this case $SP = SP_1\cdot SP_2$, and from Lemma \ref{lemmaInversion}(\ref{LITSseq}) we derive $\TypeSeq{ \T}{SP_1}{\tSeq_1}$, $ \TypeSeq{\T}{SP_2}{\tSeq_2}$, and $\tSeq = \tSeq_1\cup\tSeq_2$. By induction hypotheses on $SP_1$ and $SP_2$ we get $\TypeSeq{}{SP_1\sigma}{\tSeq_1}$ and $ \TypeSeq{}{SP_2\sigma}{\tSeq_2}$. Therefore, since $SP_1\sigma\cdot SP_2\sigma = (SP_1\cdot SP_2)\sigma$, applying the rule $\lab{TSseq}$ we conclude $\TypeSeq{}{(SP_1\cdot SP_2)\sigma}{\tSeq}$.
\end{itemize}
($\Rightarrow$) By induction on $SP$.
\begin{itemize}
 \item If $SP = \GenericVariable$, the proof follows from $\sigma\in\Sigma_\T$. If $SP = \epsilon$ or $SP = a$ we use weakening.
 \item Let $SP$ be $SP_1\cdot SP_2$. Since $(SP_1\cdot SP_2)\sigma = SP_1\sigma\cdot SP_2\sigma$,  from Lemma \ref{lemmaInversion}(\ref{LITSseq}) we derive $\tSeq = \tSeq_1\cup\tSeq_2$, $\TypeSeq{}{SP_1\sigma}{\tSeq_1}$, and $ \TypeSeq{}{SP_2\sigma}{\tSeq_2}$. By induction hypotheses we get $\TypeSeq{ \T}{SP_1}{\tSeq_1}$, and $ \TypeSeq{\T}{SP_2}{\tSeq_2}$. Applying rule ($TSseq$) we conclude $\TypeSeq{\T}{SP_1\cdot SP_2}{\tSeq}$.
\end{itemize}
\end{proof}
\end{lemma}

\begin{lemma}\label{LemmaT}
If $\sigma\in\Sigma_\T$, then $\TypePat{}{P\sigma}{\tPat}$ if and
only if $\TypePat{\T}{P}{\tPat}$.
\begin{proof}
($\Leftarrow$) By induction on $\TypePat{\T}{P}{\tPat}$. Consider
the last applied rule.
\begin{itemize}
 \item If the rule is $\lab{Tvar}$, the proof follows from $\sigma\in\Sigma_\T$. For rules $\lab{TRloc}$, $\lab{TRlocOut}$, $\lab{TRlocIn}$ the fact that $P$ is a rule implies that $P\sigma = P$ and, moreover, it is typable from the empty environment. \LB{For rule $\lab{Tseq}$ if $P$ is a sequence pattern then also $P\sigma$ is a sequence pattern, and then we can apply rule \lab{Tseq} with the empty environment.}
 \item If the rule is $\lab{Tpar}$, then $P = P_1\pipe P_2$, and from Lemma \ref{lemmaInversion}(\ref{LITpar}) we derive $\TypePat{ \T}{P_1}{\tPat_1}$, $\TypePat{\T}{P_2}{\tPat_2}$, and $\tPat = \tPat_1 \Rtunion \tPat_2$. By induction hypotheses on $P_1$ and $P_2$ we get $\TypePat{ }{P_1\sigma}{\tPat_1}$, and $ \TypePat{}{P_2\sigma}{\tPat_2}$. Therefore, since $P_1\sigma\pipe P_2\sigma = (P_1\pipe P_2)\sigma$, applying the rule $\lab{Tpar}$ we conclude $\TypePat{}{(P_1\pipe P_2)\sigma}{\tPat}$.
 \item If the rule is $\lab{Tcomp}$, then the proof is similar using Lemmas \ref{lemmaInversion}(\ref{LITcomp}) and \ref{LemmaS} for the first premise.
\end{itemize}
($\Rightarrow$) By induction on $P$.
\begin{itemize}
 \item If $P = X$, the proof follows from $\sigma\in\Sigma_\T$. \LB{If $P$ is a sequence pattern, then also $P\sigma$ is a sequence pattern, and we can apply the rule \lab{Tseq}. If $P$ is a rule, then $P = P\sigma$.}
 \item Let $P$ be $P = P_1\pipe P_2$. Since $(P_1\pipe P_2)\sigma = P_1\sigma\cdot P_2\sigma$, and the fact that $\TypePat{}{(P_1\pipe P_2)\sigma}{\tPat}$, from Lemma \ref{lemmaInversion}(\ref{LITpar}) we derive $\TypePat{}{P_1\sigma}{\tPat_1}$, $\TypePat{}{P_2\sigma}{\tPat_2}$, and $\tPat = \tPat_1 \Rtunion \tPat_2$. By induction hypotheses on $P_1$ and $P_2$ we get $\TypePat{ \T}{P_1}{\tPat_1}$ and $ \TypePat{\T}{P_2}{\tPat_2}$. Applying rule $\lab{Tpar}$ we conclude $\TypePat{\T}{(P_1\pipe P_2)}{\tPat}$.
 \item If $P = \Loop{SP} \into P'$ the proof is similar using Lemmas \ref{lemmaInversion}(\ref{LITcomp}) and \ref{LemmaS} for the first premise.
\end{itemize}
\end{proof}
\end{lemma}

\noindent
{\bf Proof of Theorem \ref{theoPreserveCorrect} (Subject Reduction)}

\noindent
By cases on the reduction rules.
\begin{description}
\item \textbf{Rule $\lab{TGR}$}

From Definition \ref{def_Tsem},
$T = E[P_1\sigma]$, $T' = E[P_2\sigma]$, and $\sigma \in
\Sigma_\T$. By hypothesis $\TypePat{}{T}\tPat$ and $\TypeGR\T{P_1
\mapsto P_2}{}$. Therefore, Lemma
\ref{lemmaHoleTyped}(\ref{lemmaHoleTyped1}) implies
$\TypePat{}{P_1\sigma}\tPat_1$ for some $\tPat_1$, and from Lemma
\ref{LemmaT} we derive $\TypePat{\T}{P_1}\tPat_1$. From
$\TypeGR\T{P_1 \mapsto P_2}{}$, Lemma
\ref{lemmaInversion}(\ref{LITRglob}) implies
$\TypePat{\T}{P_2}\tPat_2$ with $\tPat_2 \Rtsub \tPat_1$. We can
apply Lemma \ref{LemmaT} obtaining $\TypePat{}{P_2\sigma}\tPat_2$.
From Lemma \ref{lemmaHoleTyped}(\ref{lemmaHoleTyped3}) we conclude
$\TypePat{\T}{E[P_2\sigma]}\tPat'$ for some $\tPat'\Rtsub\tPat$.

\item \textbf{Rule $\lab{TLR}$}

From Definition \ref{def_Tsem}, $T = E[\Lr_1 \mapsto \Lr_2\pipe \Lr_1\sigma]$, $T' = E[\Lr_1 \mapsto \Lr_2\pipe \Lr_2\sigma]$, and $\sigma \in \Sigma_\T$. By hypothesis $\TypePat{}{T}\tPat$. Therefore, Lemma \ref{lemmaHoleTyped}(\ref{lemmaHoleTyped1}) implies $\TypePat{}{\Lr_1 \mapsto \Lr_2\pipe \Lr_1\sigma}\tPat_0$ for some $\tPat_0$. Since $\sigma \in \Sigma_\T$, Lemma \ref{LemmaT} implies $\TypePat{\T}{\Lr_1 \mapsto \Lr_2\pipe \Lr_1}\tPat_0$. From Lemma \ref{lemmaInversion}(\ref{LITpar}) we derive $\TypePat{\T}{\Lr_1 \mapsto \Lr_2}\tPat'_0$, and $\TypePat{\T}{\Lr_1}\tPat_1$ with $\tPat'_0 \Rtunion \tPat_1 = \tPat_0$. By Lemma \ref{lemmaInversion}(\ref{LITRloc}) we derive $\TypePat{\T}{\Lr_2}\tPat_2$ with $\tPat_2\Rtsub\tPat_1$. Lemma \ref{LemmaT} implies $\TypePat{}{\Lr_2\sigma}\tPat_2$, then from $\lab{Tpar}$ we derive $\TypePat{}{\Lr_1 \mapsto \Lr_2\pipe \Lr_2\sigma}\tPat_3$ with $\tPat_3 = \tPat'_0 \Rtunion \tPat_2$. Since $\tPat_3 \Rtsub \tPat_0$ we can apply Lemma \ref{lemmaHoleTyped}(\ref{lemmaHoleTyped3}) obtaining $\TypePat{}{E[\Lr_1 \mapsto \Lr_2\pipe \Lr_2\sigma]}\tPat'$ with $\tPat' \Rtsub
\tPat$.

\item \textbf{Rule $\lab{LR-Out}$}

From Definition \ref{def_Tsem}, $T = E[\Loop{S_1\sigma} \into (\up{\Lr_1}{S_1} \mapsto \up{\Lr_2}{S_2} \pipe \Lr_1\sigma \pipe T_0)]$, $T' = E[\Lr_2\sigma \pipe \Loop{S_2\sigma} \into (\up{\Lr_1}{S_1} \mapsto \up{\Lr_2}{S_2} \pipe T_0)]$, and $\sigma \in \Sigma_\T$. By hypothesis $\TypePat{}{T}\tPat$. Lemma \ref{lemmaHoleTyped}(\ref{lemmaHoleTyped1}) implies $\TypePat{}{\Loop{S_1\sigma} \into (\up{\Lr_1}{S_1} \mapsto \up{\Lr_2}{S_2} \pipe \Lr_1\sigma \pipe T_0)}\tPat_0$, and, since $\sigma \in \Sigma_\T$, Lemma \ref{LemmaT} implies $\TypePat{\T}{\Loop{S_1} \into (\up{\Lr_1}{S_1} \mapsto \up{\Lr_2}{S_2} \pipe \Lr_1 \pipe T_0)}\tPat_0$. By Lemma \ref{lemmaInversion}(\ref{LITcomp}) we get $\TypeSeq{\T}{S_1}{\tSeq_1}$, and $\TypePat{\T}{\up{\Lr_1}{S_1} \mapsto \up{\Lr_2}{S_2} \pipe \Lr_1 \pipe T_0}{\conc{\tSeq_0}{\tPat_0}}$ where $\tSeq_0\subseteq\tSeq_1$. By Lemma \ref{lemmaInversion}(\ref{LITpar}) we have $\TypePat{\T}{T_0}{\conc{\tSeq}{\tPat_1}}$, and $\TypePat{\T}{\up{\Lr_1}{S_1} \mapsto \up{\Lr_2}{S_2}}{\conc{\tSeq'}{\tPat_2}}$ for some $\conc{\tSeq}{\tPat_1} \Rtunion \conc{\tSeq'}{\tPat_2} \Rtsub
\conc{\tSeq_0}{\tPat_0}$.
Lemma \ref{lemmaInversion}(\ref{LITRlocOut}) implies $\tSeq'=\set\of$, $\TypePat\T{L_2}{\tPat_2}$ and $\TypeSeq\T{S_2}{\tSeq_2}$ where $\tSeq_1\subseteq\tSeq_2$. Since $\sigma \in \Sigma_\T$, Lemma \ref{LemmaT} implies $\TypePat{}{L_2\sigma}{\tPat_2}$, $\TypeSeq{}{S_2\sigma}{\tSeq_2}$, $\TypePat{\T}{\up{\Lr_1}{S_1} \mapsto \up{\Lr_2}{S_2}}{\conc{\set\of}{\tPat_2}}$, and $\TypePat{}{T_0}{\conc{\tSeq}{\tPat_1}}$. 
Using these premises, we apply rule $\lab{Tpar}$ deriving
$\TypePat{}{\up{\Lr_1}{S_1} \mapsto \up{\Lr_2}{S_2}\pipe
T_0}{\conc{\set\of}{\tPat_2}\Rtunion\conc{\tSeq}{\tPat_1}}$, and then
$\TypePat{}{\Loop{S_2\sigma} \into (\up{\Lr_1}{S_1} \mapsto
\up{\Lr_2}{S_2}\pipe T_0)}{\tPat_1\Rtunion\tPat_2}$ by rule
$\lab{Tcomp}$. Finally $\TypePat{}{\Lr_2\sigma \pipe
\Loop{S_2\sigma} \into (\up{\Lr_1}{S_1} \mapsto \up{\Lr_2}{S_2}
\pipe T_0)}{\tPat_1\Rtunion\tPat_2}$ by rule $\lab{Tpar}$: since
$\tPat_1\Rtunion\tPat_2 \Rtsub \tPat_0$, we can apply the Lemma
\ref{lemmaHoleTyped}(\ref{lemmaHoleTyped3}), obtaining
$\TypePat{}{E[\Lr_2\sigma \pipe \Loop{S_2\sigma} \into
(\up{\Lr_1}{S_1} \mapsto \up{\Lr_2}{S_2} \pipe T_0)]}{\tPat'}$ with
$\tPat' \Rtsub \tPat$.

 \item \textbf{Rule $\lab{LR-In}$}

 From Definition \ref{def_Tsem}, $T = E[\down{\Lr_1}{S_1} \mapsto \down{\Lr_2}{S_2} \pipe \Lr_1\sigma \pipe \Loop{S_1\sigma} \into T_0]$, $T' = E[\down{\Lr_1}{S_1} \mapsto \down{\Lr_2}{S_2} \pipe \Loop{S_2\sigma} \into (T_0 \pipe \Lr_2\sigma)]$, and $\sigma \in \Sigma_\T$. By hypothesis $\TypePat{}{T}\tPat$. Lemma \ref{lemmaHoleTyped}(\ref{lemmaHoleTyped1}) implies $\TypePat{}{\down{\Lr_1}{S_1} \mapsto \down{\Lr_2}{S_2} \pipe \Lr_1\sigma \pipe \Loop{S_1\sigma} \into T_0}\tPat_0$, and, since $\sigma \in \Sigma_\T$, Lemma \ref{LemmaT} implies $\TypePat{\T}{\down{\Lr_1}{S_1} \mapsto \down{\Lr_2}{S_2} \pipe \Lr_1 \pipe \Loop{S_1} \into T_0}\tPat_0$. By Lemma \ref{lemmaInversion}(\ref{LITpar}) we have $\TypePat{\T}{\down{\Lr_1}{S_1} \mapsto \down{\Lr_2}{S_2}}{\tPat_1}$ and $\TypePat{\T}{\Loop{S_1} \into T_0}{\tPat_2}$ for some $\tPat_1\Rtunion\tPat_2\Rtsub\tPat_0$.
 Lemma \ref{lemmaInversion}(\ref{LITRlocIn}) implies $\tPat_1 = \set\fif \Rtunion \tPat_3$, $\TypeSeq\T{S_1}{\tSeq_1}$, $\TypeSeq\T{S_2}{\tSeq_2}$, and $\TypePat\T{L_2}{\conc{\tSeq}{\tPat_3}}$, where $\tSeq\cup\tSeq_1\subseteq\tSeq_2$.  By  Lemma \ref{lemmaInversion}(\ref{LITcomp}) $\TypePat{\T}{T_0}{\conc{\tSeq'}{\tPat_2}}$ for some $\tSeq'\subseteq\tSeq_1$. Since $\sigma \in \Sigma_\T$, Lemma \ref{LemmaT} implies  $\TypeSeq{}{S_1\sigma}{\tSeq_1}$, $\TypeSeq{}{S_2\sigma}{\tSeq_2}$, $\TypePat{}{L_2\sigma}{\conc{\tSeq}{\tPat_3}}$, $\TypePat{}{\down{\Lr_1}{S_1} \mapsto \down{\Lr_2}{S_2}}{\set\fif \Rtunion \tPat_3}$, and $\TypePat{}{T_0}{\conc{\tSeq'}{\tPat_2}}$. Using these premises, we apply rule $\lab{Tpar}$ deriving $\TypePat{}{L_2\sigma\pipe T_0}{\conc{\tSeq}{\tPat_3}\Rtunion\conc{\tSeq'}{\tPat_2}}$, and then $\TypePat{}{\Loop{S_2\sigma} \into L_2\sigma\pipe T_0}{\tPat_3\Rtunion\tPat_2}$  by rule $\lab{Tcomp}$. Finally $\TypePat{}{\down{\Lr_1}{S_1} \mapsto \down{\Lr_2}{S_2} \pipe \Loop{S_2\sigma} \into L_2\sigma\pipe T_0}{\tPat_1\Rtunion\tPat_2}$, because $\tPat_1 = \set\fif\Rtunion\tPat_3$, by rule $\lab{Tpar}$: since $\tPat_1\Rtunion\tPat_2 \Rtsub \tPat_0$, we can apply the Lemma \ref{lemmaHoleTyped}(\ref{lemmaHoleTyped3}) obtaining $\TypePat{}{E[\down{\Lr_1}{S_1} \mapsto \down{\Lr_2}{S_2} \pipe \Loop{S_2\sigma} \into L_2\sigma\pipe T_0]}{\tPat'}$ for some $\tPat'\Rtsub\tPat$.

 \end{description}

\end{document}